\documentclass[12pt]{amsart}

\usepackage[utf8]{inputenc}
\usepackage{amsmath,amssymb}
\usepackage{graphicx}

\newtheorem{thm}{Theorem}
\newtheorem{lem}[thm]{Lemma}
\newtheorem{prop}[thm]{Proposition}
\newtheorem{cor}[thm]{Corollary}

\newcommand{\Z}{{\mathbb Z}} 
 
\newcommand{\R}{{\mathbb R}} 
\newcommand{\T}{{\mathbb T}} 

\newcommand{\oldspec}{{\mathcal N}_3}

\newcommand{\fixmehide}[1]{}

\title[Level repulsion for arithmetic toral point scatterers]{Level repulsion for arithmetic toral point scatterers in  dimension $3$}

\author{P\"ar Kurlberg}
\urladdr{www.math.kth.se/\~{ }kurlberg}
\address{Department of Mathematics, KTH Royal Institute of Technology,
SE-100 44 Stockholm, Sweden}
\email{kurlberg@math.kth.se}

\thanks{P.K. was partially supported by grants from 
the G\"oran Gustafsson Foundation for Research in Natural Sciences and
Medicine,
and the
Swedish Research Council (621-2011-5498, 2016-03701).}

\begin{document}
\begin{abstract}
  We show that arithmetic toral point scatterers in dimension three
  (``{\v{S}}eba billiards on $\R^3/\Z^3$'') exhibit
  strong level repulsion
  between the set of  ``new'' eigenvalues.
  More
  precisely, let $\Lambda := \{ \lambda_{1} < \lambda_{2} < \ldots \}$
  denote the ordered set of new eigenvalues. Then, given any
  $\gamma>0$, 
 $$\frac{|\{ i \leq N : \lambda_{i+1}-\lambda_{i} \leq
    \epsilon \}|}{N} = 
  O_{\gamma}(\epsilon^{4-\gamma})$$
as $N \to \infty$ (and $\epsilon>0$ small.)
  
\end{abstract}
\date{\today}
\maketitle





%


\section{Introduction}
\label{sec:introduction}
\subsection{Background}
The statistics of gaps between energy levels in the semiclassical
limit is a central problem in the theory of spectral statistics
\cite{mehta-rmt-book67,bohigas-les-houches}.
The Berry-Tabor conjecture \cite{berry-tabor} asserts that (typical)
integral systems have Poisson spacing statistics, and the
Bohigas-Giannoni-Schmit conjecture \cite{bgs} asserts that (generic)
chaotic systems should have spacing statistics given by random matrix
theory; in particular small gaps are very unlikely.

More precisely, with $\{ \lambda_{1} \leq \lambda_{2} \leq \ldots \}$
denoting the energy levels, suitably unfolded using the main term in
Weyl's law so that $|\{ i : \lambda_{i} < E \}| \sim E$ for $E$ large,
define consecutive gaps, or spacings,
$s_{i} := \lambda_{i+1}-\lambda_{i}$.  The level spacing distribution
$P(s)$, if it exists, is defined by
$$
\lim_{N \to \infty } \frac{ |\{i \leq N : s_{i} < x \}|  }{N} =
\int_{0}^{x} P(s) \, ds
$$
for all $x \geq 0$.  For Poisson spacing statistics, $P(s) = e^{-s}$,
whereas (time reversible) chaotic systems should have Gaussian
Orthogonal Ensemble (GOE) spacings, where $P(s) \approx \pi s/2$ for
$s$ small, and $P(s) \approx (\pi s/2) \exp( - \pi s^{2}/4)$ for $s$ large;
in particular, there is linear
vanishing at $s=0$ (``level repulsion'').
%

\subsubsection{Systems with intermediate statistics}
There are also
``pseudo integrable'' systems that are neither 
integrable nor chaotic.  Their spectral statistics do not fall into
the models described above and are believed to exhibit ``intermediate
statistics'', e.g. there is level repulsion as for random matrix
theory systems, whereas $P(s)$ 
has exponential tail decay similar to Poisson statistics.
%

The point scatterer, or the Laplacian perturbed by a delta potential,
for rectangular domains (i.e., in dimension $d=2$, with Dirichlet
boundary conditions) was introduced by {\v{S}}eba
\cite{seba-wave-chaos} as a model for investigating the transition
between integrability and chaos in quantum systems.  For this model
{\v{S}}eba found evidence for level repulsion of GOE type for small
gaps as well as ``wave chaos'', in particular Gaussian value distribution of
eigenfunctions.

Shigehara \cite{shigehara-wave-chaos} later pointed
out that level repulsion in dimension two
only occurs if ``strength'' of the perturbation grows logarithmically
with the eigenvalue $\lambda$. (The perturbation is formally defined
using von Neumann's theory of self adjoint extensions; in this setting
there is a one-parameter family of extensions, but any fixed parameter
choice turns out to result in a ``weak coupling limit'' with no level
repulsion, cf. \cite[Section~3]{ueberschar-point-scatterer-survey}.)
On the other hand, for dimension $d=3$, Cheon and Shigehara
\cite{shigehara-3d-torus} found GOE type level repulsion for ``fixed
strength'' perturbations for rectangular boxes, again with Dirichlet
boundary conditions and the scatterer placed at the center of the box;
placing the scatterer elsewhere appeared to weaken the repulsion.

Under certain randomness assumptions on the unperturbed spectrum,
GOE-type level repulsion and tails of Poisson type has been shown to
hold by Bogomolny, Gerland and Schmit
\cite{bogomolny-gerland-schmit-nearest-neighbor,bogomolny-gerland-schmit-singular}.
In particular, for {\v{S}}eba billiards with periodic boundary
conditions, $P(s) \sim (\pi \sqrt{3}/2) s$ (for $s$
small), whereas $P(s) \sim (1/8 \pi^3) s \log^{4} s$ in case of
Dirichlet boundary condition and generic scatterer position.
Interestingly, Berkolaiko, Bogomolny and Keating
\cite{keating-etal-star-graphs-and-seba} has shown that the
pair correlation in the former case is {\em identical} to the pair
correlation for quantum star graphs (for growing number of bonds with
generic lengths.)

\subsubsection{Toral point scatterers}
With periodic boundary conditions the location of the scatterer is
irrelevant and it is natural to consider tori.
%
In dimensions two and three, Rudnick and
Uebersch\"ar  \cite{rudnick-ueberschaer-point-scatterer-gaps} used
trace formula techniques to investigate the spacing 
distribution for toral point scatterers.
In dimension two, for a fixed self adjoint extension, resulting in a
weak coupling limit, they showed that the spectral statistics of the
new eigenvalues is the same for the old spectrum, possibly after
removing multiplicities.  (In the irrational aspect ratio case it is
believed that the old spectrum is of Poisson type, for partial results
in this direction cf.
\cite{sarnak-pair-correlation-almost-all,eskin-margulis-mozes-pair-correlation};
for the square torus $\R^2/\Z^2$ the same was shown to hold assuming
certain analogs of the Hardy-Littlewood prime $k$-tuple conjecture for
sums of two integer squares \cite{sots}.)
For $d=3$, a fixed self adjoint extension results in a strong coupling
limit, and here Rudnick-Uebersch\"ar gave  evidence for level
repulsion: the mean displacement between old and new eigenvalues was
shown to equal half the mean spacing between the old eigenvalues.
However, the method does not rule out the level spacing distribution
having (say) positive mass at $s=0$.

\subsection{Results}
The purpose of this paper is to show that there is  strong level
repulsion between the set of new eigenvalues for point scatters on
arithmetic tori in dimension three.
To state our main result we need to describe some basic properties of
the model.

\subsubsection{Toral point scatterers for arithmetic tori}
Let $\T := \R^3/(2\pi \Z^3)$ denote the standard flat 
torus in dimension three.  A point scatterer on $\T$ is formally given
by the Hamiltonian
$$
H = H_{x_{0}, \alpha} =  -\Delta + \alpha \delta_{x_{0}}, \quad \alpha \in \R
$$
where $\Delta$ is the Laplace operator acting on $L^{2}(\T)$,
$x_{0} \in \T$ is the location of the point scatterer, and $\alpha$ is
the ``strength'' of the perturbation; in the physics literature
$\alpha$ is known as the coupling constant.

A mathematically rigorous definition of $H$ can be made via von
Neumann's theory of self adjoint extensions, below we will briefly
summarize the most important properties but refer the reader to
\cite{albeverio-etal-book-solvable-models,rudnick-ueberschaer-point-scatterer-gaps,ueberschar-point-scatterer-survey}
for detailed discussions.


The spectrum of the unperturbed Laplacian on $\T$ is arithmetic in
nature, and given by $\{m \in \Z :
r_{3}(m) > 0 \}$, where
$$
r_{3}(m) := | \{ v \in \Z^3 : |v|^{2} = m \}|
$$
denotes the number of ways $m \in \Z$ can be written as a sum of three
integer squares, and the multiplicity of an eigenvalue $m$ is given by
$r_{3}(m)$.
The addition of a $\delta$-potential is a rank-$1$ perturbation of the
Laplacian, and the spectrum of $H$ consists of two kinds of
eigenvalues: ``old'' and ``new'' eigenvalues.  Namely, an eigenvalue
$m$ of the unperturbed Laplacian is also an eigenvalue of $H$, but
with multiplicity $r_{3}(m)-1$; the corresponding $H$-eigenspace is
just Laplace eigenfunctions vanishing at $x_{0}$.  For simpler
notation we shall, without loss of generality, from here on assume
that $x_{0} =0$. 

Associated to $m$
there is also a new eigenvalue $\lambda_{m}$ (of multiplicity one)
of $H$, and the set of new eigenvalues
$\{ \lambda_{m} : m \in \Z, r_{3}(m)>0 \}$ interlace between the old
eigenvalues.
%
More precisely, the corresponding new eigenfunction
$\psi_{\lambda_{m}}(x)$ is given by the
Green's function 
$$
\sum_{v \in \Z^3}  \frac{e^{2 \pi i v \cdot x}}{|v|^{2}-\lambda_{m}}
, \quad x \in \T,
$$
(in $L^{2}$ sense), with the new eigenvalue $\lambda_{m}$ being a
solution to the spectral equation
\begin{equation}
  \label{eq:spectral-equation}
G(\lambda) = \frac{1}{\nu}; \quad 
G(\lambda) :=  
\sum_{n} 
r_{3}(n) 
\left(
\frac{1}{n - \lambda } - \frac{n}{n^{2}+1}
\right),
\end{equation}
where $\nu \neq 0$ is known as the ``formal strength'' of the
perturbation in the physics literature
(cf. \cite[Eq. (4)]{shigehara-3d-torus}).

As mentioned above, the set of new eigenvalues interlace with the old spectrum
$\{ m \in \Z : r_{3}(m)>0 \}$, and it is convenient to abuse notation
and use the following labeling of the new spectrum: given $m$ such that
$r_{3}(m)>0$, let $\lambda_{m}$ denote the largest solution to
$G(\lambda) = 1/\nu$ such that $\lambda < m$.  (In
particular, $\lambda_{m}$ {\em does not} denote the $m$-th new
eigenvalue since $r_{3}(m)=0$ for a positive proportion of integers.)
For $m$ such that $r_{3}(m)>0$, let $m_{+}$ denote the
smallest integer $n > m$ such that $r_{3}(n)>0$, and define the
consecutive spacing
$$
s_{m} := \lambda_{m_{+}} - \lambda_{m}.
$$
Note that $\{ s_{m} \}_{m : r_{3}(m)>0}$ has mean $6/5$ (rather than
one), but as we are concerned mainly with the frequency of very small
spacings this shall not concern us.

%
%
%
%

\subsubsection{Statement of the main result}
We show that the { cumulant} of the nearest-neighbor distribution
essentially has fourth order vanishing near the origin, and hence
considably stronger repulsion than the quadratic order vanishing of
the cumulant in the GOE-model.

\begin{thm} 
\label{thm:main}
Given any small $\gamma>0$, we have
$$
\frac{|\{ m \leq x : r_{3}(m)>0, s_{m} < \epsilon  \} |}
{ |\{ m \leq x : r_{3}(m) >0 \}|}
= O_{\gamma}(\epsilon^{4-\gamma})
$$  
as $x \to \infty$ (and $\epsilon > 0$ small.)
\end{thm}

\subsection{Discussion}
\label{sec:discussion}

As Figures \ref{fig:l=0} and \ref{fig:l=10,20} indicate, spectral gap
statistics for $3d$ arithmetic point scatterers is clearly non-generic
since there is no mass at all in the tail.  The reason is the ``old''
spectrum being very rigid --- all positive integers $n$ except the
ones ruled out by simple congruence conditions (i.e., $n$'s of the
form $n = 4^{k} \cdot m$ for $m \equiv 7 \mod 8$) satisfy
$r_{3}(n)>0$, hence the gaps between new eigenvalues is easily seen to
be bounded above by $4$.

\begin{figure}[ht]
\centering
\includegraphics[width=.6\linewidth]{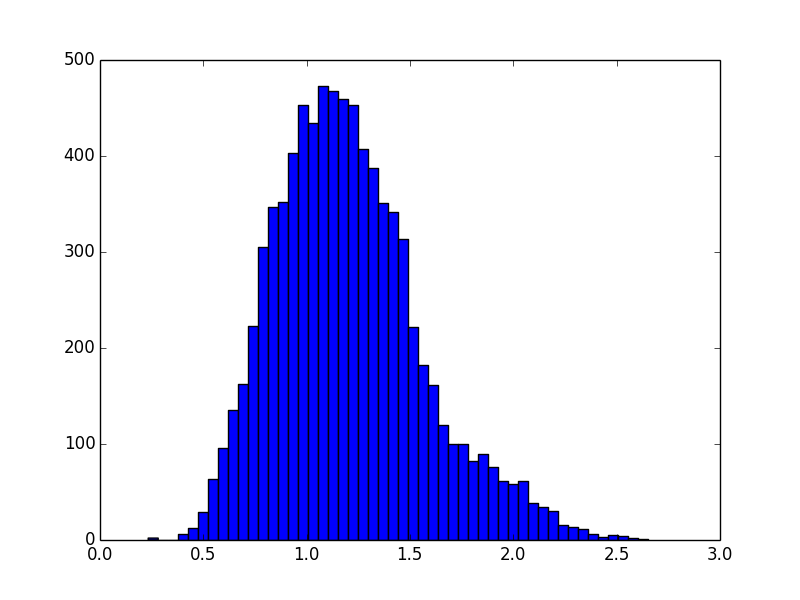}    
\caption{
Histogram illustration of the distribution of $s_{m}$, for $m \leq
10000$ (and $r_{3}(m)>0.)$}
  \label{fig:l=0}
\end{figure}

The main driving force in the fluctuations in gaps, say
between two new eigenvalues $\lambda_{m}$ and $\lambda_{m+}$, are
arithmetic in nature and mainly due to fluctuations in $r_{3}(m)$.  In
particular, the small gap repulsion is not due to lack of time
reversal symmetry, but rather due to $r_{3}(m)/\sqrt{m}$ being small
extremely rarely {\em unless} $4^{l} | m$ for some high exponent $l$.
In fact, as indicated by the plots below (as well as by the proof of
Theorem~\ref{thm:main}), small gap occurences is mainly due to
integers $m$ that are divisible by large powers of $4$.
\begin{figure}[ht]
  \centering
\includegraphics[width=.45\linewidth]{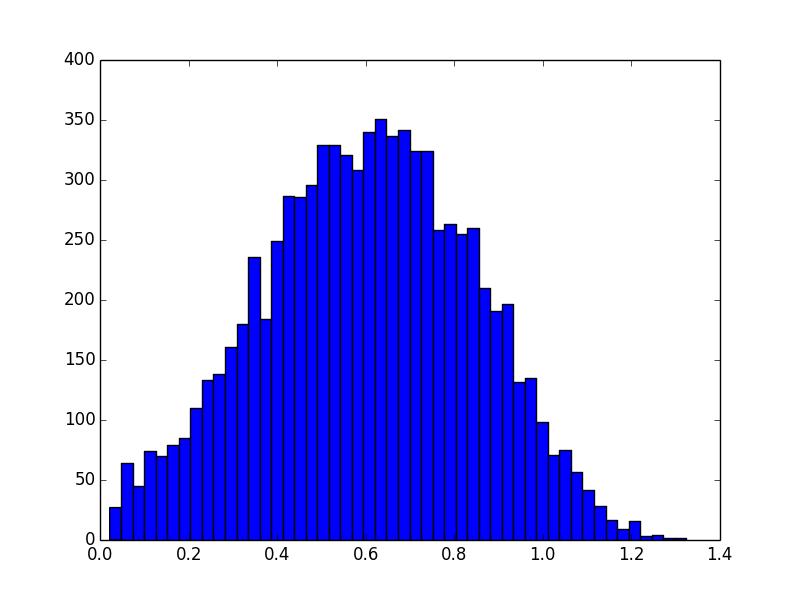}
\includegraphics[width=.45\linewidth]{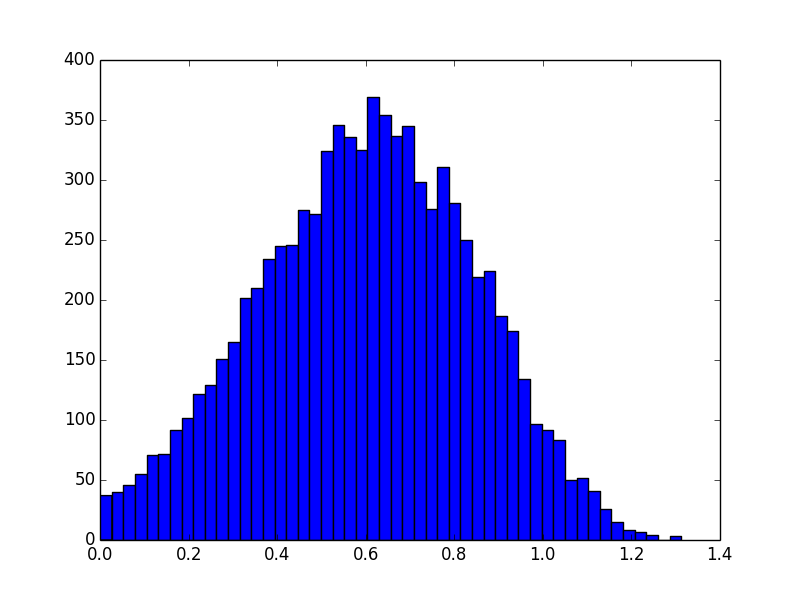}
\caption{Histogram illustration of the distribution of $s_{m}$ (for
  $m$ such that $r_{3}(m)>0$), along the progressions
  $\{m = 4^{10}\cdot k : k \leq 10000\}$ (left) and for
  $\{m = 4^{20}\cdot k : k \leq 10000\}$ (right).}
    \label{fig:l=10,20}
\end{figure}

\subsection{Acknowledgements}  The author wishes to thank Z.
Rudnick and H. Ueberschär for many helpful discussions and comments.
Part of this work was done at a working session on periodic point
scatterers during the program Periodic and Ergodic Spectral Problems
at the Newton Institute.  I would like to thank
J. Marklof for arranging the working session, and the Newton Institute
for its hospitality and excellent working conditions.

\section{Number theoretic background}

We begin by recalling some results on $r_{3}(n)$,
the number of ways to write $n \geq 0$ as a sum of three
integer squares.
A classical result of Legendre
(cf. \cite[Chapter~3.1]{grosswald-book}) asserts that
$r_{3}(n) \neq 0$ if and only if $n$ is not of the form $4^{a} (8k+7)$
for $a,k \in \Z_{\geq 0}$, and we also have
$$
r_{3}(4^{a}n) = r_{3}(n).
$$

%
We also recall a result of Heath-Brown
\cite{heath-brown-lattice-points-sphere}, namely that
\begin{equation}
  \label{eq:heath-brown-bound}
\sum_{n \leq R^{2}} r_{3}(n) =
\sum_{v \in\Z^3 : |v|^{2} \leq R^{2}} 1 
= \frac{4\pi}{3} R^{3} +  O(R^{21/16})
\end{equation}
We remark that for the main theorem any error term of the form
$ O(R^{3-\epsilon})$ and $\epsilon>0$ would suffice, but the exponent
$21/16$ is very helpful for the numerics behind Figures \ref{fig:l=0}
and \ref{fig:l=10,20}.


\subsection{Sums of three squares and values of $L$-functions}
\label{sec:sums-three-squares}

With $R_{3}(n)$ denoting the number of {\em primitive} representations
of $n$, i.e., the number of ways to write $n = x^{2}+y^{2}+z^{2}$ for
$x,y,z$ coprime, we have the following basic identity:
$$
r_{3}(n) = \sum_{d^2|n} R_{3}(n/d^{2}).
$$
The reason for introducing $R_{3}(n)$ is 
Gauss' marvelous identity (cf. \cite[Ch.~4.8]{grosswald-book})
$$
R_{3}(n) = \pi^{-1} \mu_{n} \sqrt{n}  L(1,\chi_{-4n})
$$
where $\mu_{n} = 0$ for $n \equiv 0,4,7 \mod{8}$, $\mu_{n} = 16$ for $n
\equiv 3 \mod 8$, and $\mu_{n} = 24$ for $n \equiv 1,2,5,6 \mod 8$;
$$
L(1,\chi_{-4n}) = \sum_{m=1}^{\infty} \chi_{-4n}(m)/m
$$
where $\chi_{-4n}$ is defined via the Kronecker symbol, namely
$$
\chi_{-4n}(m) := \left( \frac{-4n}{m} \right).
$$

Now, $-4n$ is not always a fundamental discriminant\footnote{
  A fundamental discriminant is an integer  $d \equiv 0,1 \mod 4$ such
  that  $d$ is squarefree if $d \equiv 1$, and $d=4m$ for $m \equiv
  2,3$ square free if $d \equiv 0 \mod 4$.}, but if $R_{3}(n) >0$
and we write
$$
n = c^{2} d
$$
where $d$ is square free and $4 \nmid c^{2}d$, then $L(1,\chi_{-4n})$
and $L(1,\chi_{-4d})$ have the same Euler factors, except at primes
$p$ dividing $4^{l} c^{2}$.
Moreover, $-4d$ is a fundamental discriminant if $n \equiv 1,2,5,6
\mod 8$. If $n \equiv 3 \mod 8$ then $-4d$ is not a fundamental
discrimant but $-d$ is, and the Euler factors of $L(1,\chi_{-4d})$ and
$L(1,\chi_{-d})$ are the same at all odd primes.  (Recall that
$n \not \equiv 0,4,7 \mod{8}$ since we assume that $R_{3}(n) >0$.)

Thus, if $4 \nmid n_{0}$  and we write $n_{0} = c^{2} d$ with $d$
square free, we note the
following useful lower bound in terms of
$L$-functions attached to 
primitive characters (associated with fundamental
discriminants)
\begin{equation}
  \label{eq:fund-disc-compare}
L(1,\chi_{-4n_{0}}) \gg
\begin{cases}
(\phi(c)/c) L(1, \chi_{-4d})    & \text{ if $n_{0} \equiv 1,2,5,6 \mod 8$,}
\\  
 (\phi(c)/c) L(1, \chi_{-d}) & \text{ if $n_{0} \equiv 3 \mod 8$,}
\end{cases}
\end{equation}
where $f(x) \gg g(x)$  means that $f(x) > c g(x)$ for some absolute
constant $c>0$.

We next show that $L(1,\chi_{-4n})$ is small very rarely.  With
$$
FD := \{ d \in \Z : d <0   \text{ and $d$ is a fundamental discriminant} \}.
$$ 
the following Proposition is an easy consequence of
\cite[Proposition~1]{granville-sound-L-one-chi-distribution}.
\begin{prop}
\label{prop:granville-sound}
There exists $c>0$ such that for $T \ge 1$, we have\footnote{We recall
  Vinogradov's ``$\ll$ notation'': $f \ll g$ is equivalent to
  $f = O(g)$.  When allowing implicit constants to depend on
a  parameter (say $\gamma$), $f \ll_{\gamma} g$ is equivalent to
  $f = O_{\gamma}(g)$.}
$$
|\{ d \leq x : -d \in FD, 
L(1,\chi_{-d}) < \frac{\pi^{2}}{6 e^{\gamma} T} \}|
\ll x \exp(-c \cdot e^{T}/T)
$$
as $x \to \infty$.
\end{prop}

For imprimitive quadratic characters we will use the following 
weaker bound.
\begin{prop} 
\label{prop:L-one-lower-bound}
  The number of integers $n \leq x$ of the form $n = 4^{l} n_{0} =
  4^{l} c^{2} d$, where $d$ is square free, $c \leq C$, $4 \nmid
  c^{2}d$, and
$$
L(1,\chi_{-4n_{0}}) \leq 1/T
$$
is
$$
\ll \frac{x}{4^{l}} \exp( - (T/\log\log C)^{4})
$$
for all integer $l \geq 0$.  
\end{prop}
\begin{proof}
  Using \eqref{eq:fund-disc-compare}, together with the bound
  $c/\phi(c) \ll \log\log c \leq \log \log C$ we find that for $c,l$
  fixed, we have (note that either $-d$ or $-4d$ is a fundamental
  discriminant)
$$
|\{ n \leq x : n = 4^{l} n_{0} = 4^{l} c^{2} d, 
L(1,\chi_{-4n_0}) \leq 1/T \}|
$$
$$
\ll
|\{ d \leq 4x/(c^{2}4^{l}) : -d \in FD, L(1,\chi_{-d}) \ll \frac{\log
  \log C}{T}  \}| 
$$
which, by Proposition~\ref{prop:granville-sound}, is
$$
\ll x/(c^{2}4^{l}) \exp( -(T/\log \log C)^{4})
$$
as $x \to \infty$.  Summing over $c \leq C$ the proof is concluded.
\end{proof}

\subsection{Estimates on moments of $r_{3}(n)$}
We recall the following bound by Barban
\cite{barban-large-sieve}.
\begin{thm}
\label{thm:barban}
For $k \in \Z^+$ we have
$$
\sum_{1 \leq d \leq x} L(1,\chi_{-d})^{k} \ll_{k} x.
$$  
\end{thm}
We can now easily deduce bounds on the (normalized) moments of
$r_{3}(n)$.
\begin{prop}
\label{prop:k-th-moment-bound}
For $k \in \Z^+$ we have
$$
\sum_{n \leq x} (r_{3}(n)/\sqrt{n})^{k}  \ll_{k} x
$$  
\end{prop}
\begin{proof}
Writing $n = 4^{l} n_{0}$ so that $4 \nmid n_{0}$ we have 
$$
r_{3}(n) = r_{3}(n_{0})
$$
and we further recall the identities
(cf. Section~\ref{sec:sums-three-squares}) 
$$
r_{3}(n) = \sum_{d^2|n} R_{3}(n/d^{2})
$$  
and
$$
R_{3}(n) = \mu_{n} \sqrt{n} L(1,\chi_{-4n}),
$$
where $\mu_{n} = 0$ if $n \equiv 0,4,7 \mod 8$; otherwise $16 \le \mu_n \leq
24$.
Thus
$$
\sum_{n\leq x} 
(r_{3}(n)/\sqrt{n})^{k} = 
\sum_{n\leq x} 
\left(
\sum_{\substack{d^2|n, d \equiv 1 \mod{2}\\4^{l}|| n}} 
\frac{R_{3}(n/(4^{l}d^{2}))}{\sqrt{n}}
\right)^{k}
$$
$$
\ll
\sum_{l : 4^{l} \leq x}
\sum_{\substack{  d^{2} \leq x/4^{l} \\ d \text{ odd} }}
\sum_{\substack{d_1,\ldots, d_{k} \\ [d_{1}, \ldots, d_{k}] = d}}
\sum_{\substack{n \leq x : d^{2} |n \\ 4^{l} || n}}
\prod_{i=1}^{k} \frac{R_{3}(n/(4^{l}d_{i}^{2}))}{\sqrt{n}} \ll
$$
$$
\sum_{l : 4^{l} \leq x}
\sum_{\substack{  d^{2} \leq x/4^{l} \\ d \text{ odd} }}
\sum_{\substack{d_1,\ldots, d_{k} \\ [d_{1}, \ldots, d_{k}] = d}}
\sum_{\substack{n_{0} \leq x/ d^{2}4^{l}}}
\prod_{i=1}^{k} \frac{R_{3}(n_{0}d^{2}/d_{i}^{2})}{\sqrt{d^{2}4^{l}n_{0}}} \ll
$$
$$
\sum_{l : 4^{l} \leq x}
\sum_{\substack{  d^{2} \leq x/4^{l} \\ d \text{ odd} }}
\sum_{\substack{d_1,\ldots, d_{k} \\ [d_{1}, \ldots, d_{k}] = d}}
\sum_{\substack{n_{0} \leq x/ d^{2}4^{l}}}
\prod_{i=1}^{k} \frac{L(1, \chi_{-4n_{0}d^{2}/d_{i}^{2}})}{d_{i}} \ll
$$
which, on noting\footnote{The Euler products for 
$L(1, \chi_{-4n_{0}d^{2}/d_{i}^{2}})$ and
$L(1, \chi_{-4n_{0}}) $ agree at all primes $p$ such that
$p\nmid d$.} that $L(1, \chi_{-4n_{0}d^{2}/d_{i}^{2}}) \ll
L(1, \chi_{-4n_{0}}) d/\phi(d) $ is
\begin{equation}
  \label{eq:noname}
\ll
\sum_{l : 4^{l} \leq x}
\sum_{\substack{  d^{2} \leq x/4^{l} \\ d \text{ odd} }}
\sum_{\substack{d_1,\ldots, d_{k} |d}}
\frac{(d/\phi(d))^{k}}{d_{1}\cdots d_{k}}
\sum_{\substack{n_{0} \leq x/ d^{2}4^{l}}}
L(1, \chi_{-4n_{0}})^{k}.
\end{equation}
By Theorem~\ref{thm:barban}, the inner sum over $n_{0}$ is $\ll_{k}
x/(d^{2}4^{l})$, and hence \eqref{eq:noname} is 
$$\ll
\sum_{l : 4^{l} \leq x}
\sum_{\substack{  d^{2} \leq x/4^{l} \\ d \text{ odd} }}
\sum_{\substack{d_1,\ldots, d_{k} |d}}
\frac{(d/\phi(d))^{k}}{d_{1}\cdots d_{k}}
\frac{x}{d^{2}4^{l}}
\ll
\sum_{l : 4^{l} \leq x}
\sum_{\substack{  d^{2} \leq x/4^{l} \\ d \text{ odd} }}
(d/\phi(d))^{2k}
\frac{x}{d^{2}4^{l}}
\ll x
$$
using that $\sum_{d_i|d} 1/d_{i} \ll d/\phi(d) \ll d^{o(1)}$.

\end{proof}

By Chebychev's inequality we immediately deduce:
\begin{cor}
Given $k \in \Z^+$, we have
$$
|\{ n \leq x : r_{3}(n)/\sqrt{n}  > T \}|
\ll_{k} x/T^{k}
$$  
\end{cor}

\section{Proof of Theorem~\ref{thm:main}}

Let $\oldspec := \{ n \in \Z : r_{3}(n)>0 \}$ denote the old spectrum.
Given $n \in \oldspec$, let $n_{+}$ denote its nearest right
neighboor in $\oldspec$, and recall our definition
$$
s_{n} := \lambda_{n_{+}} - \lambda_{n}
$$
of the nearest neighboor spacing between the two new eigenvalues
$\lambda_{n}, \lambda_{n_{+}}$.

The spectral equation (cf. (\ref{eq:spectral-equation})) is then given by
$$
\sum_{n} r_{3}(n) \left( 
\frac{1}{n-\lambda} 
- \frac{n}{n^{2}+1}
\right)
=
1/\nu
$$
where $\nu$ is constant.
For 
simpler notation we shall only treat the case of when the
R.H.S. equals zero, corresponding to $\nu = \infty$, 
but we remark that the method of proof gives the
same result when $1/\nu$ is allowed to change, sufficiently smoothly,
with $\lambda$, as long as we are in the strong coupling region
$1/\nu = O(\lambda^{1/2-\epsilon})$ for $\epsilon>0$
(cf. \cite[Section II]{shigehara-3d-torus}).  

We first show that the sum may be truncated without significantly
changing the the new eigenvalues (the proof is relegated to the
appendix.)
\begin{lem}
  \label{lem:truncation}
There exists $\alpha,\delta \in (0,1)$ such that 
$$
\sum_{n : |n-\lambda| > \lambda^{\alpha}} r_{3}(n)
\left( \frac{1}{n-\lambda} - \frac{n}{n^{2}+1} \right)
\ll \lambda^{1/2-\delta}
$$  
(in fact, using (\ref{eq:heath-brown-bound}) we may take $\alpha=1/6$.)
\end{lem}

The following simple result was used in \cite{3dscars} (cf. Lemma~7)
in order to study scarred eigenstates for arithmetic toral point
scatterers in dimensions two and three.  
\begin{lem}
\label{lem:spectral-separation}
  Assume that $A,B >0$.  Let $f(x)$ be a smooth function with $f'(x)
  \geq \epsilon >0$ for all $x \in [-1/2,1/2]$, and assume
  that $f'(x) \leq B$ for $|x| \leq 1/2$.  If the equation
$$
f(x) = A/x
$$
has two roots $x_{1} \in [-1/2,0)$ and $x_{2} \in (0,1/2]$, then
$$
|x_{2}-x_{1} |
\gg
\sqrt{A/B}.
$$
\end{lem}
%
%
%

Given $m \in \oldspec$, set $\lambda = m+\delta$, and define
$$
G_{m}(\delta) :=
\sum_{n : 0 < |n-m| \leq \sqrt{m}}
r_{3}(n) \left( 
\frac{1}{n-m-\delta} 
- \frac{n}{n^{2}+1}
\right);
$$
we can then rewrite the spectral equation as
$$
G_{m}(\delta) -1/\nu =
\frac{r_{3}(m) }{\delta} 
$$

For $|\delta| < 1/2$, we find (note that all terms are positive)
$$
0 <
G_{m}'(\delta) \ll
\sum_{n : 0 < |n-m| \leq \sqrt{m}}
\frac{r_{3}(n) }{(n-m)^{2}}
=
\sum_{0 < |h| < \sqrt{m}}
\frac{r_{3}(m+h)}{h^{2}}
$$

To apply Lemma~\ref{lem:spectral-separation}, we will need to bound
$G_{m}'$ from above.
\begin{lem}
\label{lem:B-upper-bound}
Given $k \in \Z^+$, we have
$$
\frac{1}{x}
\sum_{n \leq x} 
\left(
\sum_{0 < |h| < n^{\alpha}}
\frac{r_{3}(n+h)}{h^{2} \sqrt{n}}
\right)^{k}
\ll_{k} x
$$
and consequently 
$$
|\{ 
n \leq x : \sum_{0 < |h| < n^{\alpha}} \frac{r_{3}(n+h)}{  h^{2} \sqrt{n}} > T
\}|
\ll_{k}
x/T^{k}
$$
\end{lem}
\begin{proof}
  Expanding out the $k$-th power expression it is enough to show that
  for $0 <|h_{1}|,|h_{2}|, \ldots, |h_{k}|< x^{\alpha}$ we have
$$
\frac{1}{x}
\sum_{n \leq x } \frac{r_{3}(n+h_{1}) \cdots r_{3}(n+h_{k})}{n^{k/2}} \ll_{k} x
$$ 
and this follows from H\"older's inequality and the bound on
$\sum_{n \leq x} (r_{3}(n)/\sqrt{n})^{k}$ given by
Proposition~\ref{prop:k-th-moment-bound}.
\end{proof}

In particular, for $k \in \Z^+$, we have
\begin{equation}
  \label{eq:Gm-bound}
|\{ m \leq x : G_{m}'(\delta)/\sqrt{m}  > T
\text{ for $|\delta| \leq 1/2$} \}| \ll_{k} x/T^{k} 
\end{equation}

Now, if $s_{m} = \lambda_{m_{+}}-\lambda_{m}$ denotes the distances
between two consecutive new eigenvalues near $m \in \oldspec$, we have
one of the following: either one or both new eigenvalues lies outside
$[m-1/2, m+1/2]$ in which case $s_{m} \geq 1/2$.  In case both
eigenvalues lie in $[m-1/2, m+1/2]$, 
Lemma~\ref{lem:spectral-separation}
gives that
\begin{equation}
  \label{eq:distance-Am-Bm}
s_{m} \gg \sqrt{r_{3}(m)/G_{m}'(0)} = \sqrt{A(m)/B(m)}.
\end{equation}
where we define
$$
A = A(m) := r_{3}(m)/\sqrt{m}, \quad B = B(m) := G_{m}'(0)/\sqrt{m}
$$

\begin{prop}
For any $\gamma>0$ we have, as $x \to \infty$,
$$
|\{ n \leq x : 0< r_{3}(n)/G'_{n}(0) \leq \epsilon^{2} \}|
\ll_{\gamma}
\epsilon^{4-\gamma} \cdot x
$$  
\end{prop}

\begin{proof}
Given $n$ such that $r_{3}(n) > 0$, write $n=4^{l} n_{0} =4^{l}
 c^{2} d$, where $4 \nmid n_{0}$ and $d$ is squarefree. 
Recalling that $G_{n}'(0) = \sqrt{n} B(n)$ and 
$$
r_{3}(n) \geq R_{3}(n_{0}) \gg \frac{\sqrt{n}}{2^{l}}
L(1, \chi_{-4n_{0}})
$$
it is enough to estimate the number of $n = 4^{l} n_{0} \leq x$, for
which
\begin{equation}
  \label{eq:eqs-L-l-B}
\epsilon^{2} \geq r_{3}(n)/G_{n}'(0) \gg
\frac{L(1,\chi_{-4n_0})}{2^{l} B(4^{l}n_{0})} 
\end{equation}

The number of $n \leq x$ such that $n =
4^{l} c^{2} d$ for $c \geq \epsilon^{-4}$ is $\ll x/\epsilon^{-4} =
\epsilon^{4} x$.  We may thus assume that $c\leq \epsilon^{-4}$
in any such decomposition.

Now, given $n, \epsilon$, define $a=a(n),b=b(n)$ and $m = m(\epsilon)$
such that
$$
L(1,\chi_{-4n_{0}}) = 2^{-a}
\quad
G_{n}'(0) = 2^{b},
\quad
\epsilon = 2^{-m}.
$$
Moreover, \eqref{eq:eqs-L-l-B} is equivalent to
$
2^{-2m} \gg \frac{2^{-a}}{2^{l} 2^{b}}
$
i.e.,
$$
a+b+l \geq 2m + O(1).
$$

{\em First case: $a > \sqrt{m}$}.  Noting that $c \leq 1/\epsilon^{4}
= 2^{4m}$ implies that $\log \log c \leq  \log 4m$,
Proposition~\ref{prop:L-one-lower-bound} gives that the number of
$n \leq x$ such that $ L(1,\chi_{-4_{n_{0}}}) = 2^{-a}  \leq
2^{-\sqrt{m}}$ is
\begin{multline*}
\ll 
x \exp(-(2^{\sqrt{m}}/( \log 4m))^{4})
 \ll x \exp( -1000m) = o(\epsilon^{10} x)
\end{multline*}

{\em Second case: $ a \leq \sqrt{m}$}.  Given a large positive integer
$k \asymp 8/\gamma$, we first consider $l$ such that $l < 2m(1-\gamma)
- \sqrt{m}$. 
We then find that
$$
b \geq 2m + O(1) - a - l \geq 2m +O(1) - \sqrt{m} - 2m(1-\gamma) +
\sqrt{m} \geq 2m \gamma +O(1) 
$$
for all sufficiently large $m$.  Using the bound in
\eqref{eq:Gm-bound} (recall that $k \asymp 8/\gamma$ and  that $B(n) =
G_{n}'(0)/\sqrt{n}$),
the number of $n \leq x$ such that
$B(n) \geq 2^{\gamma m}$ is
$$
\ll_{\gamma} x/(2^{\gamma m} )^{8/\gamma} \ll_{\gamma}
 x/(2^{8m}) = O_{\gamma}( \epsilon^{8} x)
$$

Finally, the number of $n \leq x$ such that $4^{l}|n$ for some $l 
\geq 2m(1-\gamma) - \sqrt{m}$ is 
$$
\ll x/2^{2 (2m(1-\gamma) - \sqrt{m})} \ll x \epsilon^{4(1-\gamma) +o(1)}
$$
as $m \to \infty$ (or equivalently, that $\epsilon = 2^{-m}\to 0$), thereby 
concluding the proof.
\end{proof}

\appendix
\section{Spectral cutoff justification}
\label{sec:appendix}

Before proving Lemma~\ref{lem:truncation} we record complete
cancellation in a certain smooth approximation to the spectral
counting function.

\begin{lem}
\label{lem:smooth-truncation}
With $P.V. \int_{0}^{\infty} \ldots \, dt= 
\lim_{\epsilon \to 0} ( \int_{0}^{\lambda-\epsilon} \ldots +
\int_{\lambda+\epsilon}^{\infty} \ldots ) \, dt$
denoting the principal value of the integral (with respect to the
singularity at $t = \lambda$), we have
$$
P.V. \int_{0}^{\infty} 
\sqrt{t} 
\left( \frac{1}{t-\lambda} - \frac{1}{t} \right) \, dt = 0.
$$  
\end{lem}

\begin{proof}
  Using the change of variables $t = \lambda s^{2}$ 
  (taking apropriate principle values of the integrals, in particular
  for the integrals over the $s$-parameter, the principal value is
  taken with respect to the singularity at $s=1$) we find that
$$
\int_{0}^{\infty} 
\left(
  \frac{1}{\lambda-t} + \frac{1}{t}
\right) \sqrt{t} \, dt 
=
\int_{0}^{\infty} 
\left(
  \frac{1}{\lambda-\lambda s^{2}} + \frac{1}{\lambda s^{2}}
\right) \sqrt{\lambda s^{2}}  2 \lambda s \, ds
$$
$$
= 2 \sqrt{\lambda}
\int_{0}^{\infty} 
\left(
  \frac{1}{1- s^{2}} + \frac{1}{ s^{2}}
\right) s^{2} \, ds
=
2 \sqrt{\lambda}
\int_{0}^{\infty} 
  \frac{2}{1- s^{2}} 
\, ds  = 0.
$$
\end{proof}

\begin{proof}[Proof of Lemma~\ref{lem:truncation}]
%
%
%
%
%
As $n/(n^{2}+1)-1/n = O(1/n^{3})$, we find that
$$
\sum_{n = 1}^{\infty} r_{3}(n) 
\left(
\frac{1}{n-\lambda}   - \frac{n}{n^{2}+1}
\right)
= 
\sum_{n = 1}^{\infty} r_{3}(n) g(n) + O(1)
$$
if we define
$$
g(n) := \frac{1}{n-\lambda} - \frac{1}{n}
$$
We now study sums 
$$
\sum_{n \in I} r_{3}(n) g(n)
$$
for various intervals $I = (A,B]$.   More precisely,
define $I_{1} := (1,\lambda - \lambda^{\alpha}]$,
$I_{2} = (\lambda - \lambda^{\alpha}, \lambda+\lambda^{\alpha}]$,
and let $I_{3} := (\lambda+\lambda^{\alpha}, \infty)$, for $\alpha \in
(0,1)$.  It is our 
goal to show that for $k=1,3$, we have
$$
\sum_{n \in I_{k}}  r_{3}(n) g(n) \ll \lambda^{1/2 - \delta}
$$
for some small $\delta>0$, if $\alpha$ is chosen appropriately.

With
$$
R(t) := \sum_{n \le t } r_{3}(n)
$$
we have (cf.  \eqref{eq:heath-brown-bound})
$$
R(t) = \frac{4 \pi}{3} t^{3/2} + O(t^{21/32})
=: M(t) + E(t).
$$

By Abel's summation formula,
$$
\sum_{ A < n \le B}
r_{3}(n) g(n)
=
R(B)g(b) - R(A)g(A) 
-
\int_{A}^{B} R(t) g'(t) \, dt,
$$
and the contribution from the smooth main term $M(t)$ is then
$$
=
M(B)g(b) - M(A)g(A) 
-
\int_{A}^{B} M(t) g'(t) \, dt
$$
which, on using partial integration, equals (after taking appropriate
principal values)
  \begin{multline*}
M(B)g(b) - M(A)g(A) 
-
\left( M(B)g(b) - M(A)g(A) 
+
\int_{A}^{B} m(t) g(t) \, dt \right)
\\ =
\int_{A}^{B} m(t) g(t) \, dt 
  \end{multline*}
where $m(t) = M'(t) = c \sqrt{t}$ for some $c>0$.

Thus, taking principal values as appropriate,   and using
Lemma~\ref{lem:smooth-truncation} 
$$
\int_{I_{1} \cup I_{3}}
m(t) g(t) \, dt  
=
- \int_{0}^{1} m(t) g(t) \, dt - \int_{I_{2}} m(t) g(t) \, dt  =
O(1) - \int_{I_{2}} m(t) g(t) \, dt  
$$
and it is enough to bound 
\begin{equation}
  \label{eq:PV-bound}
\int_{I_{2}} 
\sqrt{t}
\left( \frac{1}{t-\lambda} - \frac{1}{t} \right) \, dt 
=
\int_{-\lambda^{\alpha}}^{\lambda^{\alpha}}
\sqrt{\lambda+s}
\left( \frac{1}{s} - \frac{1}{\lambda+s} \right) \, ds.
\end{equation}
Since $\sqrt{\lambda+s} = \sqrt{\lambda}\cdot( 1 + O(s/\lambda)$)
we find that (\ref{eq:PV-bound}) equals
$$
\int_{-\lambda^{\alpha}}^{\lambda^{\alpha}} \frac{\lambda^{1/2}}{s} \,
ds
+
\int_{-\lambda^{\alpha}}^{\lambda^{\alpha}}
\frac{\lambda^{1/2}s}{\lambda s} \, ds
+ O(\lambda^{\alpha-1/2}) = 
0+  O(\lambda^{\alpha-1/2})  = 
 O(\lambda^{\alpha-1/2}).
$$

Further, the contribution from the error term $E(t)$ equals
\begin{multline*}
E(B)g(B) - E(A)g(A) 
-
\int_{A}^{B} E(t) g'(t) \, dt
\\=
B^{21/32} \left( \frac{1}{B-\lambda} - \frac{1}{B}  \right)
-
\left(A^{21/32} \left( \frac{1}{A-\lambda} - \frac{1}{A}  \right)
\right)
\\
+O 
\left(
\int_{A}^{B} t^{21/32}
\left( -\frac{1}{(t-\lambda)^{2}} + \frac{1}{t^{2}}  \right) \, dt
\right).
\end{multline*}
For the integral over the first interval $I_{1}$ we have $A=1$,
$B=\lambda-\lambda^{\alpha}$ and thus
$$
\int_{I_{1}}
\ll
\lambda^{21/32-\alpha} 
+ O(1) +
\lambda^{21/32-\alpha} +
O(1)
\ll \lambda^{21/32-\alpha} 
$$
(to see that
$\int_{1}^{\lambda-\lambda^{\alpha{}}} 
\frac{t^{21/32} }{(t-\lambda)^{2}}  dt \ll \lambda^{21/32-\alpha}$,
consider $\int_1^{\lambda/2}$ and
$\int_{\lambda/2}^{\lambda-\lambda^{\alpha}}$ separately.)
Similarly, the contribution to $\int_{I_{3}}$, from terms involving
$E(t)$, is
$$
\ll
\lambda^{21/32-\alpha}
$$
(to see that
$\int_{\lambda+\lambda^{\alpha{}}}^{\infty} \frac{t^{21/32}
}{(t-\lambda)^{2}} dt \ll \lambda^{21/32-\alpha}$,
consider $\int_{\lambda+\lambda^{\alpha{}}}^{2\lambda}$ and
$\int_{2\lambda}^{\infty}$ separately.)

In conclusion, we find that
$$
\sum_{ n \in I_{1} \cup I_{2}}
r_{3}(n) g(n)
\ll \lambda^{21/32-\alpha} = \lambda^{1/2-\delta}
$$
for $\delta>0$ if we take $\alpha = 1/6$ (note that $21/32 < 2/3$.)

\end{proof}

\bibliographystyle{abbrv} 
\bibliography{mybib,mypapers}

\end{document}